\documentclass[a4paper]{article} 

\usepackage{amsmath, amssymb, amsthm}
\usepackage{mathrsfs} 
\usepackage{dsfont} 
\usepackage{qip} 
\usepackage[section]{thmenvironments} 

\usepackage{underscore} 
\usepackage{authblk} 
\usepackage{appendix} 
\usepackage{microtype} 
\usepackage[margin=10pt,font=small,labelfont=bf,labelsep=endash]{caption}

\usepackage{hyperref}  
\usepackage{cite} 
\usepackage{hyperlinks} 
\usepackage{eprint} 

\hypersetup{
  pdftitle={Key recycling in authentication},
  pdfauthor={De, Portmann, Vidick, Renner},
  pdfsubject={Quantum-proof extractors},
  pdfstartview={FitH}
}

\title{Trevisan's extractor in the presence of quantum side information}

\author[1]{Anindya De\thanks{Supported by the Berkeley fellowship for
    Graduate study and NSF-CCF-1017403.}}

\author[2,3]{Christopher
  Portmann\thanks{Supported by the Swiss
    National Science Foundation (via grant Nos.~200021-119868 and
    200020-135048 and the National Centre of Competence in Research
    `Quantum Science and Technology') and the European Research
    Council --- ERC (grant no.~258932).}\thanks{Supported in part by
    Vienna Science and Technology Fund (WWTF) through project
    ICT10-067 (HiPANQ).}}

\author[4]{Thomas Vidick\thanks{Supported by the National Science
    Foundation under Grant No.~0844626.}}

\author[2]{Renato Renner$^\dag$}

\affil[1]{Computer Science Division, University of California,
  Berkeley, CA, USA. \mailto{anindya@cs.berkeley.edu}}

\affil[2]{Institute for Theoretical Physics, ETH Zurich, 8093
  Zurich, Switzerland. \texttt{\{\href{mailto:chportma@phys.ethz.ch}{chportma},\href{mailto:renner@phys.ethz.ch}{renner}\}@phys.ethz.ch}}

\affil[3]{Group of Applied Physics, University of Geneva, 1211
  Geneva, Switzerland.}

\affil[4]{Computer Science and Artificial Intelligence
  Laboratory, Massachusetts Institute of Technology, Cambridge, MA,
  USA. \mailto{vidick@csail.mit.edu}.}

\date{\today}

\begin{document}

\maketitle

\begin{abstract}
  Randomness extraction involves the processing of purely classical
  information and is therefore usually studied in the framework of
  classical probability theory. However, such a classical treatment is
  generally too restrictive for applications where side information
  about the values taken by classical random variables may be
  represented by the state of a quantum system. This is particularly
  relevant in the context of cryptography, where an adversary may make
  use of quantum devices.  Here, we show that the well known
  construction paradigm for extractors proposed by Trevisan is sound
  in the presence of quantum side information.

  We exploit the modularity of this paradigm to give several concrete
  extractor constructions, which, e.g., extract all the conditional
  (smooth) min-entropy of the source using a seed of length
  poly-logarithmic in the input, or only require the seed to be weakly
  random.
\end{abstract}


\section{Introduction}
\label{sec:intro}

Randomness extraction is the art of generating (almost) uniform
randomness from any weakly random source $X$.  More precisely, a
\emph{randomness extractor} (or, simply \emph{extractor}) is a
function $\Ext$ that takes as input $X$ together with a uniformly
distributed (and usually short) string $Y$, called the \emph{seed},
and outputs a string $Z$. One then requires $Z$ to be almost uniformly
distributed whenever the min-entropy of $X$ is larger than some
threshold $k$, i.e.,
\begin{align} \label{eq:reqstand}
  \HminOp(X) \geq k  \, \implies \, Z \coloneqq \Ext(X,Y) \text{ statistically close to uniform}.
\end{align}
The min-entropy of a random variable $X$ is directly related to the
probability of correctly guessing the value of $X$ using an optimal
strategy: $2^{-\Hmin{X}} = \max_x P_X(x)$. Hence \critref{eq:reqstand}
can be interpreted operationally: if the maximum probability of
successfully guessing the input of the extractor, $X$, is sufficiently
low then its output is statistically close to uniform.

The randomness of a value $X$ always depends on the information one
has about it, in the following called \emph{side information}.  In
cryptography, for instance, a key is supposed to be uniformly random
from the point of view of an adversary, who may have access to
messages exchanged by the honest parties, which we would therefore
consider as side information. Here, extractors are typically used for
\emph{privacy amplification}~\cite{BBR88,BBCM95}, i.e., to turn a
partially secure raw key (about which the adversary may have
non-trivial information) into a perfectly secure key. We thus demand
that the extractor output be uniform with respect to the side
information held by the adversary.  Another example is
\emph{randomness recycling} in a computation, which can be done using
extractors~\cite{IZ89}. The aim is that the recycled randomness be
independent of the outputs of previous computations, which are
therefore considered as side information.

In the following, we make side information explicit and denote it by
$E$. The notions of randomness we are going to use, such as the
\emph{guessing probability}, \emph{min-entropy} or the
\emph{uniformity} of a random variable, must then be defined with
respect to $E$.  We can naturally
reformulate~\critref{eq:reqstand} as
\begin{align} \label{eq:reqside}
  \HminOp(X|E) \geq k  \, \implies \, Z \coloneqq \Ext(X,Y) \ & \text{statistically close to uniform}\\
  &\text{conditioned on $E$,}\notag
\end{align}
where $\Hmin{X|E}$ is the conditional min-entropy, formally defined in
\secref{subsec:min-entropy}. This conditioning naturally extends the
operational interpretation of the min-entropy to scenarios with
explicit side information, i.e., $2^{-\Hmin{X|E}}$ is the maximum
probability of correctly guessing $X$, given access to side
information $E$~\cite{KRS09}.

Interestingly, the relationship between the two
Criteria~\eqref{eq:reqstand} and \eqref{eq:reqside} depends on the
physical nature of the side information $E$, i.e., whether $E$ is
represented by the state of a classical or a quantum system. In the
case of purely classical side information, $E$ may be modeled as a
random variable and it is known that the two criteria are essentially
equivalent (see \lemref{lem:KT08} for a precise statement). But in the
general case where $E$ is a quantum system, \critref{eq:reqside} is
\emph{strictly stronger} than~\eqref{eq:reqstand}: it was shown
in~\cite{GKKRW07} that there exist extractors that
fulfill~\eqref{eq:reqstand} but for which~\eqref{eq:reqside} fails
(see also~\cite{KR07} for a discussion).

Since our world is inherently non-classical, it is of particular
importance that~\eqref{eq:reqside} rather than the
weaker~\critref{eq:reqstand} be taken as the relevant criterion for
the definition of extractors. In cryptography, for instance, there is
generally nothing that prevents an adversary from holding quantum side
information. In fact, even if a cryptographic scheme is purely
classical, an adversary may acquire information using a non-classical
attack strategy. Hence, when using extractors for privacy
amplification, \critref{eq:reqstand} does not generally imply
security. A similar situation may arise in the context of randomness
recycling. If we run a (simulation of) a quantum system $E$ using
randomness $X$, approximately $\Hmin{X|E}$ bits of $X$ can be
reused. If we now, in an attempt to recycle the randomness, apply a
function $\Ext$ which fulfills \eqref{eq:reqstand} but
not~\eqref{eq:reqside}, the output $Z$ may still be correlated to the
system $E$.

It is known that the conditional min-entropy accurately characterizes
the maximum amount of uniform randomness that can be extracted from
$X$ while being independent from $E$. (More precisely, the
\emph{smooth conditional min-entropy}, an entropy measure derived from
$\Hmin{X|E}$ by maximizing the latter over all states in an
$\eps$-neighborhood, is an upper bound on the amount of uniform
randomness that can be extracted; see ~\secref{subsec:min-entropy}
and~\cite{Ren05} for details). In other words, the characterization of
extractors in terms of $\Hmin{X|E}$ is essentially optimal, and one
may thus argue that \critref{eq:reqside} is indeed the correct
definition for randomness extraction (see
also~\cite{Ren05,KR07,KT08}). In this work, we follow this line of
argument and call an extractor \emph{quantum-proof} if it
satisfies~\critref{eq:reqside} (see \secref{subsec:extractors.defs}).

We note that there have been alternative proposals in the literature
for defining extractors in the context of quantum side information,
which do however not satisfy the above optimality condition. One
prominent example is the bounded storage model
(see~\secref{subsec:local}), where the (quantum) side information $E$ is
characterized by the number of qubits, $H_0(E)$, required to store
it. In this model, the entropy $\Hmin{X|E}$ of a source $X$
conditioned on $E$ is lower-bounded by $\Hmin{X}-H_0(E)$.  However,
this characterization of side information is strictly weaker than that
using $\Hmin{X|E}$: there are sources $X$ and nontrivial side
information $E$ such that $\Hmin{X}-H_0(E) \ll
\Hmin{X|E}$.\footnote{This can easily be seen by considering the
  following example. Let $X$ be uniformly distributed on $\{0,1\}^n$
  and $E$ be $X$ with each bit flipped with constant probability $\eps
  < 1/2$. Then $\Hmin{X|E} = \Theta(n)$, but $\Hmin{X}-H_0(E) = 0$.}
In particular, even if an extractor can provably extract
$\Hmin{X}-H_0(E)$ bits of uniform (with respect to $E$) randomness
from a source $X$, we do not know whether the same extractor can
attain the optimal $\Hmin{X|E}$ bits. Note also that the same
considerations apply to the purely classical case. In fact, no recent
work defines classical extractors for randomness sources with side
information stored in bounded classical memories.\footnote{Restricting
  the class of randomness sources further than by bounding their
  min-entropy can have advantages. For example, if we consider only
  bit-fixing sources, or sources generated by a random walk on a
  Markov chain, then the extractor can be deterministic. (See
  \cite{Sha02} for a brief overview of restricted families of sources
  studied in the literature.) There is however no known advantage
  (e.g., in terms of seed length) in considering only input sources
  with side information stored in a memory of bounded size, whether it
  is classical or quantum.}

Finally we remark that the increased generality attained by the notion
of quan\-tum-proof extractors used here is crucial for applications. For
example in quantum key distribution, where extractors are used for
privacy amplification~\cite{Ren05}, it is generally impossible to
bound the adversary's memory size.

\subsection{Related results}

In the standard literature on randomness extraction, constructions of
extractors are usually shown to fulfill \critref{eq:reqstand}, for
certain values of the threshold $k$ (see~\cite{Zuc90} as well
as~\cite{Sha02} for an overview). However, only a few constructions
have been shown to fulfill \critref{eq:reqside} with arbitrary
quantum side information $E$. Among them is two-universal
hashing~\cite{Ren05,TSSR10}, constructions based on the
sample-and-hash approach~\cite{KR07}, as well as all extractors with
one-bit output~\cite{KT08}.

Recently, Ta-Shma~\cite{Ta09} studied Trevisan's~\cite{Tre01}
construction of extractors in the bounded quantum storage model.  The
result was a breakthrough because it, for the first time, implied the
existence of quantum-proof extractors requiring only short seeds
(logarithmic in the input length). Unfortunately, Ta-Shma's result is
proved in the bounded quantum storage model. More precisely, he
requires the output length to be much smaller than the min-entropy of
the original data: it scales as $(\Hmin{X}/H_0(E))^{1/c}$, where $c>1$
is a constant.

Subsequent to this work, Ben-Aroya and Ta-Shma~\cite{BT12} showed how
two versions of Trevisan's extractor, shown quantum-proof in this
paper, can be combined to extract a constant fraction of the
min-entropy of an $n$-bit source with a seed of length $O(\log n)$,
when $\Hmin{X|E} > n/2$. This is better than the straightforward
application of Trevisan's extractor analyzed here, which requires
$O(\log^2 n)$ bits of seed for the same output size (but works for any
$\Hmin{X|E}$).

\subsection{Our contribution}
In this work, we show that the performance of Trevisan's extractor
does not suffer in the presence of quantum side information.  This
improves on the best previously known result~\cite{Ta09} in two major
ways. First, we prove our results in the most general model, where the
min-entropy of the source is measured relative to quantum side
information (\critref{eq:reqside}). Second, we show that the output
length of the extractor can be close to the optimal conditional
min-entropy $\Hmin{X|E}$ (see \corref{cor:optimalentropyloss2} for the
exact parameters).\footnote{In the conference version of this
  paper~\cite{DV10}, two of us showed that a similar result could be
  obtained in the more restricted bounded-storage model.} This
provides the first proof of soundness for an extractor with
poly-logarithmic seed meeting \critref{eq:reqside} in the presence of
arbitrary quantum side information.

More generally, we show that a whole class of extractors is
quantum-proof. It has been observed, by, e.g., Lu~\cite{Lu04} and
Vadhan~\cite{Vad04}, that Trevisan's extractor~\cite{Tre01} (and
variations of it, such as \cite{RRV02}) can be seen as a concatenation
of the outputs of a one-bit extractor with different pseudo-random
seeds. Since the proof of the extractor property is independent of the
type of the underlying one-bit extractor (and to some extent the
construction of the pseudo-random seeds), our result is valid for a
generic scheme (defined in \secref{subsec:genericscheme},
\defref{def:genericscheme}). We find that the performance of this
generic scheme in the context of quantum side information
(\secref{subsubsec:security.uniform}, \thmref{thm:security}) is roughly
equivalent to the (known) case of purely classical side
information~\cite{RRV02}.

In practical situations where quantum-proof extractors are used, e.g.,
privacy amplification in quantum key distribution~\cite{Ren05}, the
players do not necessarily have access to a uniform source of
randomness. We therefore analyze separately the situation where the
seed is only weakly random, and show that Trevisan's extractor is
quantum-proof in that setting as well
(\secref{subsubsec:security.non-uniform},
\thmref{thm:security.non-uniform}).

By ``plugging'' various one-bit extractors and pseudo-random seeds
into the ge\-ner\-ic scheme, we obtain different final constructions,
optimized for different needs, e.g., maximizing the output length,
minimizing the seed, or using a non-uniform seed. In
\tableref{tab:constructions} we give a brief overview of the final
constructions proposed.

\begin{table}[tb]
  \begin{tabular}{| l | c | c | c | p{2.2cm} |} \hline & Min-entropy &
    Output length & Seed length & Note \\ \hline
    \shortcorref{cor:optimalentropyloss2} & any $k$ & $m = k - 4 \log
    1/\eps$ & $d = O(\log^3 n)$ & optimized output length \\[3pt]
    \hline \shortcorref{cor:logseed} & $k = n^\alpha$ & $m = n^{\alpha -
      \gamma}$ & $d = O(\log n)$ & optimized seed length \\ \hline
    \shortcorref{cor:localextractor} & $k = \alpha n$ & $m =
    (\alpha - \gamma)n$ & $d = O(\log^2 n)$ & local extractor\\
    \hline \shortcorref{cor:weak-random-seed} &$k = n^\alpha$ & $m =
    n^{\alpha - \gamma}$ & $d = O(\log n)$ & seed with
    min-entropy $\beta d$ \\ \hline
  \end{tabular}
  \caption[Overview of concrete constructions]{Plugging various weak
    designs and $1$-bit extractors in Trevisan's construction, we
    obtain these concrete extractors. Here $n$ is the input
    length, $\eps = \poly{1/n}$ the error, $\alpha$ and $\gamma$ are
    arbitrary constants such that $0 < \gamma < \alpha \leq 1$, and
    $\frac{1}{2} < \beta < 1$ is a specific constant.}
  \label{tab:constructions}
\end{table}

\subsection{Proof technique}

The proof proceeds by contradiction. We first assume that a player
holding the side information $E$ can distinguish the output from
uniform with probability greater than $\eps$. We then show that such a
player can reconstruct the input $X$ with high probability, which
means that $X$ must have low min-entropy ($\Hmin{X|E} < k$). Taking
the contrapositive proves that the extractor is sound.

Trevisan~\cite{Tre01} originally proved the soundness of his
extractor this way. His construction starts by encoding the source $X$
using a list-decodable code $C$. The output of
the extractor then consists of certain bits of $C(X)$, which are
specified by the seed and a construction called a (weak)
design~\cite{NW94,RRV02}. (See \secref{subsec:genericscheme} for a
precise description of Trevisan's extractor.) His proof can then be
broken down in two steps. He first shows that a player who can
distinguish the output from uniform can guess a random bit of
$C(X)$. In the second step, he shows that such a player can
reconstruct $X$.

Proving the soundness of Trevisan's extractor in the quantum min-entropy
framework requires some important changes. 
In order to better explain these new elements, it will be
useful to first give a brief overview of the main steps that go into
Ta-Shma's proof~\cite{Ta09}. For the sake of contradiction, assume that there is a
test $T$ which performs a measurement on the side information $E$ in
order to distinguish the output from uniform with advantage $\eps$. Using a
standard hybrid argument, along with properties of the (weak) design,
one can then construct a new test $T'$ (using a little extra classical
advice about $X$) which predicts a random bit of $C(X)$ with
probability $\frac{1}{2}+\frac{\eps}{m}$, where $m$ is the number of
output bits. Further, $T'$ makes exactly \emph{one} query to $T$.

The proof in~\cite{Ta09}
proceeds by showing how from such a test, one can construct another
test $T''$ which predicts any bit of $X$ with probability $0.99$ and
queries $T'$ at most $q=(m/\eps)^c$ times ($c=15$ for the code
in~\cite{Ta09}). This gives a random access code (RAC)~\cite{ANTV99}
for $X$; however, since it requires $q$ queries to the side
information $E$, the no-cloning theorem forces us to see it as
querying a single system of length $qH_0(E)$ (recall that Ta-Shma's
result was proved in the bounded storage model, where one bounds the
information provided by $E$ by its number of qubits $H_0(E)$). Finally, 
using a new bound on the dimension of RACs~\cite{Ta09}, one
finds that $\Hmin{X} \gtrsim m^c H_0(E)$, hence $m \lesssim
(\Hmin{X}/H_0(E))^{1/c}$, where for simplicity we have taken the error
$\eps$ to be a constant.

Our proof improves upon Ta-Shma's through two major changes. 
First, we model the side information $E$
explicitly, instead of viewing it as an oracle which one queries. Indeed, the
measurement performed by the test $T'$ to predict the bits of $C(X)$
will be different from the measurement performed by $T''$ to
reconstruct $X$, and this cannot be captured by the
``oracle side-information'' model of Ta-Shma. We
thus show (in \secref{subsec:security}, \propref{prop:guessing-cx}) that
if the output of the extractor can be distinguished from uniform with
probability $\frac{1}{2}+\eps$ by a player holding the side
information $E$, then the bits of $C(X)$ can be guessed with
probability $\frac{1}{2}+\frac{\eps}{m}$ by a player holding $E$ and
some extra small classical information $G$.

Second, we depart from the reconstruction paradigm at the heart of the
second half of the proof of both Trevisan's and Ta-Shma's
results. Instead of defining explicitly the measurement and
computation necessary to reconstruct $X$, we use the fact that for any
list-decodable code $C : \{0,1\}^n \to \{0,1\}^{\bar{n}}$, the
function
\begin{align*}
    C' : \{0,1\}^n \times [\bar{n}] & \to \{0,1\} \\
    (x,i) & \mapsto C(x)_i
\end{align*}
is a one-bit extractor according to \critref{eq:reqstand}
(see \appendixref{sec:codesRextractors} for more details). It was
however proved by K\"onig and Terhal~\cite{KT08}, that in the one-bit
setting the more general \critref{eq:reqside} is essentially
equivalent to the usual \critref{eq:reqstand}. This result lets us
conclude directly that the input $X$ must have low min-entropy
relative to the quantum side information $E$.

This proof structure results in a very modular extractor construction
paradigm, which allows arbitrary one-bit extractors and pseudo-random
seeds to be plugged in, producing many different final
constructions, some of which are given in \tableref{tab:constructions}
and detailed in \secref{sec:plugging}.

\subsection{Organization of the paper}

We first define the necessary
technical tools in \secref{sec:preliminaries}, in particular the
conditional min-entropy. In \secref{sec:extractors} we give formal
definitions of extractors and discuss how much randomness can be
extracted from a given source. \secref{sec:main} contains the
description of Trevisan's extractor construction paradigm and our main
result: a proof that this construction paradigm is sound in the presence of
quantum side information, in the cases of both uniform and weakly
random seeds. Then in \secref{sec:plugging} we plug into Trevisan's
construction various one-bit extractors and pseudo-random seed
constructions, resulting in various different extractors. For example,
\secref{subsec:optimalentropyloss} contains a construction which is
nearly optimal in the amount of randomness extracted (which is
identical to the best known bound in the classical case~\cite{RRV02}
for Trevisan's extractor), and \secref{subsec:weak-random-seed} gives an
extractor which is still sound if there is a small linear entropy loss
in the seed. Finally, in \secref{sec:others}, we give a brief outlook
on further work. In particular, we mention a few classical results
which modify and improve Trevisan's extractor, but for which the
soundness in the presence of quantum side information does not seem to
follow immediately from this work.

The appendix contains many technical sections and lemmas which are not
essential for understanding Trevisan's extractor, but are nonetheless
an important part of the construction and
proof. \appendixref{sec:more} develops a bit more the general theory
of extractors: it contains two subsections which, respectively, define
extractors for weakly random seeds and show how to compose extractors
to obtain more randomness from the same
source. In \appendixref{sec:lemmas} we state several technical lemmas:
min-entropy chain rules and the details of the reduction from
Trevisan's construction to the underlying one-bit extractor. Finally,
in \appendixref{sec:codesRextractors} we give a proof that
list-decodable codes are one-bit extractors.

\section{Technical preliminaries}
\label{sec:preliminaries}

\subsection{Notation}
\label{subsec:notation}

We write $[N]$ for the set of integers $\{1,\dotsc,N\}$. If $x \in
\{0,1\}^n$ is a string of length $n$, $i \in [n]$ an integer, and $S
\subseteq [n]$ a set of integers, we write $x_i$ for the \ith{i} bit
of $x$, and $x_S$ for the string formed by the bits of $x$ at the
positions given by the elements of $S$.

\hilbert\ always denotes a finite-dimensional Hilbert space. We denote
by $\po{}$ the set of positive semi-definite operators on \hilbert. We
define the set of normalized quantum states $\no{} \coloneqq \{ \rho
\in \po{} : \tr \rho = 1\}$ and the set of sub-normalized quantum
states $\sno{} \coloneqq \{ \rho \in \po{}: \tr \rho \leq 1\}$.

We write $\hilbert_{AB} = \hilbert_A \tensor \hilbert_B$ for a
bipartite quantum system and $\rho_{AB} \in \po{AB}$
for a bipartite quantum state. $\rho_A = \trace[B]{\rho_{AB}}$ and
$\rho_B = \trace[A]{\rho_{AB}}$ denote the corresponding reduced
density operators.

If a classical random variable $X$ takes the value $x \in \cX$
with probability $p_x$, it can be represented by the state $\rho_X =
\sum_{x \in X} p_x \proj{x}$, where $\{\ket{x}\}_{x \in \cX}$
is an orthonormal basis of a Hilbert space $\hilbert_X$. If the
classical system $X$ is part of a composite system $XB$, any state of
that composite system can be written as $\rho_{XB} = \sum_{x \in
  \cX} p_x \proj{x} \tensor \rho^x_B$.

$\trnorm{\cdot}$ denotes the trace norm and is defined by
$\trnorm{A} \coloneqq \tr \sqrt{\hconj{A}A}$.


\subsection{Min-entropy}
\label{subsec:min-entropy}

To measure how much randomness a source contains and can be extracted,
we need to use the \emph{smooth conditional min-entropy}. This entropy
measure was first defined by Renner~\cite{Ren05}, and represents the
optimal measure for randomness extraction in the sense that it is
always possible to extract that amount of almost-uniform randomness
from a source, but never more. Before defining this notion, we first
state a \emph{non-smooth} version.

\begin{deff}[conditional min-entropy~\cite{Ren05}]
  \label{def:min-entropy}
  Let $\rho_{AB} \in \sno{AB}$. The
  \emph{min-entropy} of $A$ conditioned on $B$ is defined as
  \begin{equation*}
    \Hmin[\rho]{A|B} \coloneqq \max \{\lambda
    \in \R :  \exists \sigma_B \in \no{B} \,\,\mathrm{s.t.}\,\, 2^{-\lambda} \1_A \tensor \sigma_B \geq \rho_{AB}\}.
  \end{equation*}
\end{deff}

We will often drop the subscript $\rho$ when there is no doubt about what
underlying state is meant.

This definition has a simple operational interpretation when the
first system is classical, which is the case we consider. K\"onig et
al.~\cite{KRS09} showed that for a state $\rho_{XB} = \sum_{x \in
  \cX} p_x \proj{x} \tensor \rho^x_B$ classical on
$X$, \begin{equation} \label{eq:hmin=pguess} \Hmin[\rho]{X|B} = - \log
  \pguess[\rho]{X|B},\end{equation} where $\pguess{X|B}$ is the
maximum probability of guessing $X$ given $B$, namely
\begin{equation*}
  \pguess[\rho]{X|B} \coloneqq
  \max_{\{E^x_B\}_{x \in \cX}} \left( \sum_{x \in \cX}
    p_x \trace{E^x_B\rho^x_B} \right),
\end{equation*}
where the maximum is taken over all POVMs $\{E^x_B\}_{x \in \cX}$ on $B$.
If the system $B$ is empty, then the min-entropy of $X$ reduces to the
Renyi entropy of order infinity, $\Hmin{X} = - \log \max_{x \in \cX}
p_x$ (sometimes written $H_\infty(X)$). In this case the connection to
the guessing probability is particularly obvious: when no side
information is available, the best guess we can make is simply the
value $x \in \cX$ with highest probability.

The \emph{smooth} min-entropy then consists in maximizing the
min-entropy over all sub-normalized states $\eps$-close to the actual
state $\rho_{XB}$ of the system considered. Thus by introducing an
extra error $\eps$, we have a state with potentially much more
entropy. (See \secref{subsec:more-randomness} for more details.)





\begin{deff}[smooth min-entropy~\cite{Ren05,TCR09}]
  \label{def:smooth-min-entropy}
  Let $\eps \geq 0$ and $\rho_{AB} \in \sno{AB}$, then the
  \emph{$\eps$-smooth min-entropy} of $A$ conditioned on $B$ is defined as
  \begin{equation*}
    \HminSmooth[\rho]{\eps}{A|B} \coloneqq \max_{\tilde{\rho}_{AB}
      \in \cB^\eps(\rho_{AB})} \Hmin[\tilde{\rho}]{A|B},
  \end{equation*}
  where $\cB^\eps(\rho_{AB}) \subseteq \sno{AB}$ is a ball of
  sub-normalized states of radius $\eps$ around
  $\rho_{AB}$.\footnote{The distance measure used in this definition
    is the \emph{purified distance}~\cite{TCR09}, $P(\rho,\sigma)
    \coloneqq \sqrt{1 - F(\rho,\sigma)^2}$, where $F(\cdot,\cdot)$ is
    the fidelity. The only property of the purified distance we need
    in this work is that it upper bounds the trace distance, i.e.,
    $P(\rho,\sigma) \geq \frac{1}{2}\trnorm{\rho - \sigma}$. We refer
    to \cite{TCR09} for a formal definition of the purified distance
    (and fidelity) on sub-normalized states and a discussion of its
    advantages.}
\end{deff}

\section{Extractors}
\label{sec:extractors}

\subsection{Extractors, side information, and privacy amplification}
\label{subsec:extractors.defs}

An extractor $\Ext: \{0,1\}^n \times \{0,1\}^d \to \{0,1\}^m$ is a
function which takes a weak source of randomness $X$ and a uniformly
random, short seed $Y$, and produces some output $\Ext(X,Y)$, which is
almost uniform. The extractor is said to be strong, if the output is
approximately independent of the seed.

\begin{deff}[strong extractor~\cite{NZ96}]
  \label{def:extractor}
  A function $\Ext: \{0,1\}^n \times \{0,1\}^d \to \{0,1\}^m$ is a
  \emph{$(k,\eps)$-strong extractor with uniform seed}, if for all
  distributions $X$ with min-entropy $\Hmin{X} \geq k$ and a uniform seed $Y$, we
  have\footnote{A more standard classical notation would be
    $\frac{1}{2} \left\| \Ext(X,Y) \circ Y - U_m \circ Y \right\| \leq
    \eps$, where the distance metric is the variational
    distance. However, since classical random variables can be
    represented by quantum states diagonal in the computational basis,
    and the trace distance reduces to the variational distance, we use
    the quantum notation for compatibility with the rest of this
    work.} \[\frac{1}{2} \trnorm{ \rho_{\Ext(X,Y)Y} - \rho_{U_m}
    \tensor \rho_Y} \leq \eps, \] where $\rho_{U_m}$ is the fully
  mixed state on a system of dimension $2^m$.
\end{deff}

Using the connection between min-entropy and guessing probability
(\eqnref{eq:hmin=pguess}), a $(k,\eps)$-strong extractor can be
seen as a function which guarantees that if the guessing probability 
of $X$ is not too high ($\pguess{X} \leq 2^{-k}$), then it
produces a random variable which is approximately uniform and
independent from the seed $Y$.

As discussed in the introduction, we consider here a more general
situation involving side information, denoted by $E$, which may be
represented by the state of a quantum system.  A function $\Ext$ is
then an extractor if, when the probability of guessing $X$ \emph{given
  $E$} is not too high, $\Ext$ can produce a random variable
$\Ext(X,Y)$ which is approximately uniform and independent from the
seed $Y$ and the side information $E$. Equivalently, one may think of
a \emph{privacy amplification} scenario~\cite{BBR88,BBCM95}, where $E$
is the information available to an adversary and where the goal is to
turn weakly secret data $X$ into a \emph{secret} key $\Ext(X,Y)$,
where the seed $Y$ is assumed to be public. (In typical key agreement
protocols, the seed is chosen by the legitimate parties and exchanged
over public channels.)

The following definition covers the general situation where the side
information $E$ may be represented quantum-mechanically. The case of
purely classical side information is then formulated as a restriction
on the nature of $E$.

\begin{deff}[quantum-proof strong
  extractor~\protect{\cite[Section 2.6]{KR07}}]
  \label{def:extractorwithadversary}
  A function $\Ext: \{0,1\}^n \times \{0,1\}^d \to \{0,1\}^m$ is a
  \emph{quantum-proof} (or simply \emph{quantum})
  \emph{$(k,\eps)$-strong extractor with uniform seed}, if for all
  states $\rho_{XE}$ classical on $X$ with $\Hmin[\rho]{X|E} \geq k$,
  and for a uniform seed $Y$, we
  have\footnoteremember{fn:TSSRdef}{\cite{TSSR10} substitute $\exists
    \sigma_{YE} \text{\ s.t.\ } \frac{1}{2} \trnorm{
      \rho_{\Ext(X,Y)YE} - \rho_{U_m} \tensor \sigma_{YE}} \leq \eps$
    for \eqnref{eq:extractorwithadversary}. This results in a weaker
    definition, which does not offer the same composability
    guarantees. In particular, \lemref{lem:composition} does not hold
    with the same parameters when extractors are defined as in
    \cite{TSSR10}.}
\begin{equation} \label{eq:extractorwithadversary} \frac{1}{2} \trnorm{
    \rho_{\Ext(X,Y)YE} - \rho_{U_m} \tensor \rho_Y \tensor \rho_E}
  \leq \eps, \end{equation} where $\rho_{U_m}$ is the fully mixed state on a
  system of dimension $2^m$.

  The function $\Ext$ is a \emph{classical-proof $(k,\eps)$-strong
    extractor with uniform seed} if the same holds with the system $E$
  restricted to classical states.
\end{deff}

It turns out that if the system $E$ is restricted to
classical information about $X$, then this definition is essentially
equivalent to the conventional \defref{def:extractor}.

\begin{lem}[\protect{\cite[Section 2.5]{KR07},\cite[Proposition 1]{KT08}}]\label{lem:KT08}
  Any $(k,\eps)$-strong extractor is a classical-proof $(k+\log
  1/\eps,2\eps)$-strong extractor.
\end{lem}

However, if the system $E$ is quantum, this does not necessarily
hold. Gavinsky et al.~\cite{GKKRW07} give an example of a
$(k,\eps)$-strong extractor, which breaks down in the presence of
quantum side information, even when $\Hmin{X|E}$ is significantly
larger than $k$.

\begin{rem} \label{rem:weak-seed} In this section we defined
  extractors with a uniform seed, as this is the most common way of
  defining them. Instead one could use a seed which is only weakly
  random, but require it to have a min-entropy larger than a given
  threshold, $\Hmin{Y} \geq s$. The seed must still be independent
  from the input and the side information. Since having access to a
  uniform seed is often an unrealistic assumption, it is much more
  useful for practical applications to define and prove the soundness
  of extractors with a weakly random seed. We redefine extractors
  formally this way in \appendixref{subsec:weak-seed}, and show in
  \secref{subsubsec:security.non-uniform} that Trevisan's extractor is
  still quantum-proof in this setting.

  All the considerations of this section, in particular
  \lemref{lem:KT08} and the gap between classical and quantum
  side-information, also apply if the seed is only weakly random. In
  the following, when we talk about a strong extractors without
  specifying the nature of the seed, we are referring to both uniform
  seeded and weakly random seeded extractors.
\end{rem}

\subsection{Extracting more randomness}
\label{subsec:more-randomness}

Radhakrishnan and Ta-Shma~\cite{RT00} have shown that a
$(k,\eps)$-strong extractor $\Ext: \{0,1\}^n \times \{0,1\}^d \to
\{0,1\}^m$ will necessarily have \begin{equation} \label{eq:RT00} m
  \leq k - 2 \log 1/\eps + O(1).\end{equation} However, in some
situations we can extract much more randomness than the
min-entropy. For example, let $X$ be distributed on $\{0,1\}^n$ with
$\Pr[X = x_0] = 1/n$ and for all $x \neq x_0$, $\Pr[X = x] =
\frac{n-1}{n(2^n - 1)}$. We have $\Hmin{X} = \log n$, so using a
$(\log n, 1/n)$-strong extractor we could obtain at most $\log n$ bits
of randomness. But $X$ is already $1/n$-close to uniform, since
$\frac{1}{2} \trnorm{\rho_X - \rho_{U_n}} \leq \frac{1}{n}$. So we
already have $n$ bits of nearly uniform randomness, exponentially more
than the min-entropy suggests.

In the case of quantum extractors, similar examples can be found,
e.g., in \cite[Remark~22]{TCR09}. However, an upper bound on the
extractable randomness can be obtained by replacing the min-entropy by
the \emph{smooth} min-entropy (\defref{def:smooth-min-entropy}). More
precisely, the total number of $\eps$-uniform bits that can be
extracted in the presence of side information $E$ can never exceed
$\HminSmooth{\eps}{X|E}$~\cite[Section 5.6]{Ren05}.

Conversely, the next lemma implies that an extractor which is known to
extract $m$ bits from any source such that $\Hmin{X|E}\geq k$ can in
fact extract the same number of bits, albeit with a slightly larger
error, from sources which only satisfy $\HminSmooth{\eps'}{X|E}\geq
k$, a much weaker requirement in some cases.

\begin{lem}
  \label{lem:smooth-min-entropy}
  If $\Ext: \{0,1\}^n \times \{0,1\}^d \to \{0,1\}^m$ is a
  quantum-proof $(k,\eps)$-strong extractor, then for any state
  $\rho_{XE}$ and any $\eps' > 0$ with $\HminSmooth[\rho]{\eps'}{X|E}
  \geq k$, \[\frac{1}{2} \trnorm{\rho_{\Ext(X,Y)YE} - \rho_{U_m} \tensor
    \rho_Y \tensor \rho_E} \leq \eps + 2\eps'.\]
\end{lem}

\begin{proof}
  Let $\tilde{\rho}_{XE}$ be the state $\eps'$-close to $\rho_{XE}$
  for which $\Hmin[\tilde{\rho}]{X|E}$ reaches its maximum. Then
  \begin{align*}
    & \frac{1}{2} \trnorm{\rho_{\Ext(X,Y)YE} - \rho_{U_m} \tensor \rho_Y \tensor \rho_E} \\
    & \qquad  \leq \frac{1}{2} \trnorm{\rho_{\Ext(X,Y)YE} -
      \tilde{\rho}_{\Ext(X,Y)YE}} + \frac{1}{2}
    \trnorm{\tilde{\rho}_{\Ext(X,Y)YE} - \rho_{U_m} \tensor \rho_Y \tensor
      \tilde{\rho}_E} \\ & \qquad  \qquad + \frac{1}{2} \trnorm{
      \rho_{U_m} \tensor \rho_Y \tensor \tilde{\rho}_E - \rho_{U_m}
      \tensor \rho_Y \tensor \rho_E} \\
    & \qquad  \leq \frac{1}{2} \trnorm{\tilde{\rho}_{\Ext(X,Y)YE} - \rho_{U_m} \tensor
      \rho_Y \tensor \tilde{\rho}_E} + \trnorm{\rho_{XE} -
      \tilde{\rho}_{XE}} \\
    & \qquad  \leq \eps + 2 \eps'.  \end{align*}
  In the second inequality above we used twice the fact that a
  trace-preserving quantum operation can only decrease the trace
  distance. And in the last line we used the fact that the purified
  distance --- used in the
  smooth min-entropy definition (\defref{def:smooth-min-entropy}) ---
  upper bounds the trace distance.  \end{proof}



\begin{rem}
\label{rem:entropy-loss}
Since a $(k,\eps)$-strong extractor can be applied to any source with
smooth min-entropy $\HminSmooth{\eps'}{X|E} \geq k$, we can measure
the entropy loss of the extractor --- namely how much entropy was not
extracted --- with \[ \Delta \coloneqq k - m, \] where $m$ is the size
of the output. From \eqnref{eq:RT00} we know that an extractor has
optimal entropy loss if $\Delta = 2 \log 1/\eps + O(1)$.\end{rem}

\section{Constructing \emph{m}-bit extractors from one-bit extractors
  and weak designs}
\label{sec:main}

In this section we prove our main result: we show that Trevisan's
extractor paradigm~\cite{Tre01} --- which shows how to construct an
$m$-bit extractor from any (classical) $1$-bit strong
extractor --- is sound in the presence of quantum side information.

This construction paradigm can be seen as a derandomization of the
simple concatenation of the outputs of a $1$-bit extractor applied $m$
times to the same input with different (independent) seeds. The
construction with independent seeds needs a total seed of length
$d=mt$, where $t$ is the length of the seed of the $1$-bit
extractor. Trevisan~\cite{Tre01} shows how to do this using only $d =
\poly{t, \log m}$ bits of seed, and proves it is sound when no side
information is present.\footnote{Trevisan's original paper does not
  explicitly define his extractor as a pseudo-random concatenation of
  a $1$-bit extractor. It has however been noted in, e.g.,
  \cite{Lu04,Vad04}, that this is basically what Trevisan's extractor
  does.} We combine a combinatorial construction called weak designs
by Raz et al.~\cite{RRV02}, which they use to improve Trevisan's
extractor, and a previous observation by two of the
authors~\cite{DV10}, that since $1$-bit extractors were shown to be
quantum-proof by by K\"onig and Terhal~\cite{KT08}, Trevisan's
extractor is also quantum-proof.

This results in a generic scheme, which can be based on any weak design and
$1$-bit strong extractor. We define it in \secref{subsec:genericscheme}, then
prove bounds on the min-entropy and error in \secref{subsec:security}.

\subsection{Description of Trevisan's construction}
\label{subsec:genericscheme}

In order to shorten the seed while still outputting $m$ bits, in 
Trevisan's extractor construction paradigm the seed is treated
as a string of length $d <mt$, which is then split in $m$
overlapping blocks of $t$ bits, each of which is used as a (different)
seed for the $1$-bit extractor. Let $y \in \{0,1\}^d$ be the total
seed. To specify the seeds for each application of the $1$-bit
extractor we need $m$ sets $S_1,\cdots,S_m \subset [d]$ of size $|S_i|
= t$ for all $i$. The seeds for the different runs of the $1$-bit
extractor are then given by $y_{S_i}$, namely the bits of $y$ at the
positions specified by the elements of $S_i$.

The seeds for the different outputs of the $1$-bit extractor must
however be nearly independent. To achieve this, Nisan and
Wigderson~\cite{NW94} proposed to minimize the overlap $|S_i \cap
S_j|$ between the sets, and Trevisan used this idea in his original
work~\cite{Tre01}. Raz et al.~\cite{RRV02} improved this, showing that
it is sufficient for these sets to meet the conditions of a \emph{weak
  design}.\footnote{The second condition of the weak design was
  originally defined as $\sum_{j = 1}^{i-1} 2^{|S_j \cap S_i|} \leq
  r(m-1)$. We prefer to use the version of \cite{HR03}, since it
  simplifies the notation without changing the design constructions.}

\begin{deff}[weak design~\protect{\cite{RRV02}}]
  \label{def:weakdesign} A family of sets $S_1,\dotsc,S_m \subset [d]$
  is a \emph{weak $(t,r)$-design} if
  \begin{enumerate}
    \item For all $i$, $|S_i| = t$.
    \item For all $i$,  $\sum_{j = 1}^{i-1} 2^{|S_j \cap S_i|} \leq
      rm$.
  \end{enumerate}  
\end{deff}

We can now describe Trevisan's generic extractor construction.

\begin{deff}[Trevisan's extractor~\cite{Tre01}] \label{def:genericscheme} For a
  one-bit extractor $C : \{0,1\}^n \times \{0,1\}^t \to \{0,1\}$,
  which uses a (not necessarily uniform) seed of length $t$, and for a
  weak $(t,r)$-design $S_1,\dotsc,S_m \subset [d]$, we define the
  $m$-bit extractor $\Ext_C : \{0,1\}^n \times \{0,1\}^d \to
  \{0,1\}^m$ as \[\Ext_C(x,y) \coloneqq C(x,y_{S_1}) \dotsb
  C(x,y_{S_m}).\] 
\end{deff}

\begin{rem}
  \label{rem:seedsize}
  The length of the seed of the extractor $\Ext_C$ is $d$, one of the
  parameters of the weak design, which in turn depends on $t$, the
  size of the seed of the $1$-bit extractor $C$. In
  \secref{sec:plugging} we will give concrete instantiations of weak
  designs and $1$-bit extractors, achieving various entropy losses and
  seed sizes. The size of the seed will always be $d = \poly{\log n}$,
  if the error is $\eps = \poly{1/n}$. For example, to achieve a near
  optimal entropy loss (\secref{subsec:optimalentropyloss}), we need
  $d = O(t^2 \log m)$ and $t = O(\log n)$, hence $d = O(\log^3 n)$.
\end{rem}

\subsection{Analysis}
\label{subsec:security}

We now prove that the extractor defined in the previous section is a
quantum-proof strong extractor. The first step follows the structure
of the classical proof~\cite{Tre01,RRV02}. We show that a player
holding the side information and who can distinguish the output of the
extractor $\Ext_C$ from uniform can --- given a little extra
information --- distinguish the output of the underlying $1$-bit
extractor $C$ from uniform. This is summed up in the following
proposition:

\begin{prop} \label{prop:guessing-cx} Let $X$ be a classical random
  variable correlated to some quantum system $E$, let $Y$ be a
  (not necessarily uniform) seed, independent from $XE$, and let
  \begin{equation} \label{eq:guessing-cx.prop.1} 
    \trnorm{\rho_{\Ext_C(X,Y)YE}
      - \rho_{U_m} \tensor \rho_Y \tensor \rho_E} > \eps,
  \end{equation} where $\Ext_C$ is the extractor from
  \defref{def:genericscheme}. Then there exists a fixed partition of the
  seed $Y$ in two substrings $V$ and $W$, and a classical random
  variable $G$, such that $G$ has size $H_0(G) \leq rm$,
  where $r$ is one of the parameters of the weak design
  (\defref{def:weakdesign}), $V \leftrightarrow W
  \leftrightarrow G$ form a Markov chain,\footnote{Three random
    variables are said to form a Markov chain  $X \leftrightarrow Y
    \leftrightarrow Z$ if for all $x,y,z$ we have $P_{Z|YX}(z|y,x) =
    P_{Z|Y}(z|y)$, or equivalently $P_{ZX|Y}(z,x|y) =
    P_{Z|Y}(z|y)  P_{X|Y}(x|y)$.} and
  \begin{equation} \label{eq:guessing-cx.prop.2}
     \trnorm{\rho_{C(X,V)VWGE} - \rho_{U_1} \tensor
         \rho_{VWGE}} > \frac{\eps}{m}.
   \end{equation}
\end{prop}

We provide a proof of \propref{prop:guessing-cx}
in \appendixref{sec:security.reduction}, where it is restated as
\propref{prop:guessing-cx-bis}.\footnote{Note that Ta-Shma~\cite{Ta09} has
  already implicitly proved that this proposition must hold in the
  presence of quantum side information, by arguing that the side information
  can be viewed as an oracle. The present statement is a strict
  generalization of that reasoning, which allows conditional
  min-entropy as well as non-uniform seeds to be used.}

For readers familiar with Trevisan's scheme~\cite{Tre01,RRV02}, we
briefly sketch the correspondence between the variables of
\propref{prop:guessing-cx} and quantities analyzed in Trevisan's
construction. Trevisan's proof proceeds by assuming by contradiction
that there exists a player, holding $E$, who can distinguish
between the output of the extractor and the uniform distribution
(\eqnref{eq:guessing-cx.prop.1}). Part of the seed is then fixed (this
corresponds to $W$ in the above statement) and some classical advice
is taken (this corresponds to $G$ in the above statement) to construct
another player who can distinguish a specific bit of the output
from uniform. But since a specific bit of Trevisan's extractor is just
the underlying $1$-bit extractor applied to a substring of the seed
($V$ in the above statement), this new player (who holds $WGE$) can
distinguish the output of the $1$-bit extractor from uniform
(\eqnref{eq:guessing-cx.prop.2}).

In the classical case \propref{prop:guessing-cx} would be sufficient
to prove the soundness of Trevisan's scheme, since it shows that if
a player can distinguish $\Ext_C$ from uniform, then he can
distinguish $C$ from uniform given a few extra advice bits, which
contradicts the assumption that $C$ is an extractor.\footnote{In the
  classical case, \cite{Tre01,RRV02} still show that a player who
  can distinguish $C(X,V)$ from uniform can reconstruct $X$ with high
  probability. But this is nothing else than proving that $C$ is an
  extractor.}  But since our assumption is that the underlying $1$-bit
extractor is only classical-proof, we still need to show that the
quantum player who can distinguish $C(X,V)$ from uniform is not
more powerful than a classical player, and so if he can distinguish
the output of $C$ from uniform, so can a classical player. This has
already been done by K\"onig and Terhal~\cite{KT08}, who show that
$1$-bit extractors are quantum-proof.

\begin{thm}[\protect{\cite[Theorem III.1]{KT08}}]
  \label{thm:1-bit-against-Q}
  Let $C : \{0,1\}^n \times \{0,1\}^t \to \{0,1\}$ be a
  $(k,\eps)$-strong extractor. Then $C$ is a quantum-proof $(k + \log
  1/\eps,3\sqrt{\eps})$-strong extractor.\footnote{This result holds
    whether the seed is uniform or not.}
\end{thm}

We now need to put \propref{prop:guessing-cx} and
\thmref{thm:1-bit-against-Q} together to prove that Trevisan's
extractor is quantum-proof. The cases of uniform and weak random seeds
differ somewhat in the details. We therefore give two separate proofs
for these two cases in \secref{subsubsec:security.uniform} and
\secref{subsubsec:security.non-uniform}.


\subsubsection{Uniform seed}
\label{subsubsec:security.uniform}

We show that Trevisan's extractor is a quantum-proof strong extractor with uniform seed
with the following parameters.

\begin{thm}
  \label{thm:security}
  Let $C : \{0,1\}^n \times \{0,1\}^t \to \{0,1\}$ be a
  $(k,\eps)$-strong extractor with uniform seed and $S_1,\dotsc,S_m
  \subset [d]$ a weak $(t,r)$-design. Then the extractor given in
  \defref{def:genericscheme}, $\Ext_C : \{0,1\}^n \times \{0,1\}^d \to
  \{0,1\}^m$, is a quantum-proof $(k + rm + \log 1/\eps,
  3m\sqrt{\eps})$-strong extractor.
\end{thm}


\begin{proof}
  In \propref{prop:guessing-cx}, if the seed $Y$ is uniform, then $V$
  is independent from $W$ and hence by the Markov chain property from
  $G$ as well, so \eqnref{eq:guessing-cx.prop.2} can be rewritten as
  \[ \trnorm{\rho_{C(X,V)VWGE} - \rho_{U_1} \tensor \rho_V \tensor
    \rho_{WGE}} > \frac{\eps}{m}, \] which corresponds to the exact
  criterion of the definition of a quantum-proof extractor.

  Let $C$ be a $(k,\eps)$-strong extractor with uniform seed, and
  assume that a player holds a system $E$ such that
 \[\trnorm{\rho_{\Ext_C(X,Y)YE}
    - \rho_{U_m} \tensor \rho_Y \tensor \rho_E} > 3m\sqrt{\eps}.\] Then by
  \propref{prop:guessing-cx} and because $Y$ is uniform, we know that
  there exists a classical system $G$ with $H_0(G) \leq
  rm$, and a partition of $Y$ in $V$ and $W$, such that,
  \begin{equation} \label{eq:security.thm.1} \trnorm{\rho_{C(X,V)VWGE}
      - \rho_{U_1} \tensor \rho_V \tensor \rho_{WGE}} > 3
    \sqrt{\eps}.\end{equation}

  Since $C$ is a $(k,\eps)$-strong extractor, we know from
  \thmref{thm:1-bit-against-Q} that we must have $\Hmin{X|WGE} < k +
  \log 1/\eps$ for \eqnref{eq:security.thm.1} to hold. Hence by
  \lemref{lem:min-entropy.rule.3}, $\Hmin{X|E} = \Hmin{X|WE} \leq
  \Hmin{X|WGE} +H_0(G) < k + rm + \log 1/\eps$.  \end{proof}

\subsubsection{Weak random seed}
\label{subsubsec:security.non-uniform}

We show that Trevisan's extractor is a quantum-proof strong
extractor with weak random seed, with the following parameters.

\begin{thm}
  \label{thm:security.non-uniform}
  Let $C : \{0,1\}^n \times \{0,1\}^t \to \{0,1\}$ be a
  $(k,\eps)$-strong extractor with an $s$-bit seed --- i.e., the seed
  needs at least $s$ bits of min-entropy --- and $S_1,\dotsc,S_m
  \subset [d]$ a weak $(t,r)$-design. Then the extractor given in
  \defref{def:genericscheme}, $\Ext_C : \{0,1\}^n \times \{0,1\}^d \to
  \{0,1\}^m$, is a quantum-proof $(k + rm + \log 1/\eps,
  6m\sqrt{\eps})$-strong extractor for any seed with min-entropy
  $d-(t-s-\log \frac{1}{3\sqrt{\eps}})$.
\end{thm}

The main difference between this proof and that of
\thmref{thm:security}, is that since the seed $Y$ is not uniform in
\propref{prop:guessing-cx}, the substring $W$ of the seed not used by
the $1$-bit extractor $C$ is correlated to the seed $V$ of $C$, and
acts as classical side information about the seed. To handle this, we
show in \lemref{lem:non-uniform-seed} that with probability $1-\eps$
over the values of $W$, $V$ still contains a lot of min-entropy,
roughly $s'-d'$, where $d'$ is the length of $W$ and $s'$ is the
min-entropy of $Y$. And hence a player holding $WGE$ can
distinguish the output of $C$ from uniform, even though the seed has
enough min-entropy.

\begin{proof}
  Let $C$ be a $(k,\eps)$-strong extractor with $s$ bits of
  min-entropy in the seed, and assume that a player holds a system
  $E$ such that
  \[\trnorm{\rho_{\Ext_C(X,Y)YE} - \rho_{U_m} \tensor \rho_Y \tensor
    \rho_E} > 6m\sqrt{\eps}.\] Then by
  \propref{prop:guessing-cx} we have
  \begin{equation} \label{eq:security.thm.2} \trnorm{\rho_{C(X,V)VWGE}
      - \rho_{U_1} \tensor \rho_{VWGE } } > 6 \sqrt{\eps}. \end{equation}

  Since this player has classical side-information $W$ about the seed
  $V$, we need an extra step to handle
  it. \lemref{lem:non-uniform-seed} tells us that from
  \eqnref{eq:security.thm.2} and because by
  \thmref{thm:1-bit-against-Q}, $C$ is a quantum $(k+\log
  1/\eps,3\sqrt{\eps})$-strong extractor, we must have either for some
  $w$, $\Hmin{X|GEW=w} < k + \log 1/\eps$
  and hence \begin{multline*}\Hmin{X|E} = \Hmin{X|EW=w} \\
    \leq \Hmin{X|GEW=w} +H_0(G) < k + rm + \log 1/\eps,\end{multline*}
  or $\Hmin{V|W} < s + \log \frac{1}{3\sqrt{\eps}}$, from which we
  obtain using \lemref{lem:min-entropy.rule.1}
  \[\Hmin{Y}
  \leq \Hmin{V|W} + H_0(W) < s + \log \frac{1}{3\sqrt{\eps}} + d
  -t. \qedhere \] \end{proof}

\section{Concrete constructions}
\label{sec:plugging}

Depending on what goal has been set --- e.g., maximize the output,
minimize the seed length --- different $1$-bit extractors and weak
designs will be needed. In this section we give a few examples of what
can be done, by taking various classical extractors and designs, and
plugging them into \thmref{thm:security} (or
\thmref{thm:security.non-uniform}), to obtain bounds on the seed size
and entropy loss in the presence of quantum side information.

The results are usually given using the $O$-notation. This is always
meant with respect to all the free variables, e.g., $O(1)$ is a
constant independent of the input length $n$, the output length $m$,
and the error $\eps$. Likewise, $o(1)$ goes to $0$ for both $n$ and
$m$ large.

We first consider the problem of extracting all the min-entropy of the
source in \secref{subsec:optimalentropyloss}. This was achieved in the
classical case by Raz et al.~\cite{RRV02}, so we use the same
$1$-bit extractor and weak design as them.

In \secref{subsec:logseed} we give a scheme which uses a seed of length
$d = O(\log n)$, but can only extract part of the entropy. This is
also based on Raz et al.~\cite{RRV02} in the classical case.

In \secref{subsec:local} we combine an extractor and design which are
locally computable (from Vadhan~\cite{Vad04} and Hartman and
Raz~\cite{HR03} respectively), to produce a quantum $m$-bit extractor,
such that each bit of the output depends on only $O(\log(m/\eps))$
bits of the input.

And finally in \secref{subsec:weak-random-seed} we use a $1$-bit
extractor from Raz~\cite{Raz05}, which only requires a weakly random
seed, resulting in a quantum $m$-bit extractor,
which also works with a weakly random seed.

These constructions are summarized in \tableref{tab:constructions} on
\pref{tab:constructions}.

\subsection{Near optimal entropy loss}
\label{subsec:optimalentropyloss}

To achieve a near optimal entropy loss we need to combine a $1$-bit
extractor with near optimal entropy loss and a weak $(t,1)$-design. We use
the same extractor and design as Raz et al.~\cite{RRV02} to do so.

\begin{lem}[\protect{\cite[Lemma 17]{RRV02}}\footnote{Hartman and
    Raz~\cite{HR03} give a more efficient construction of this lemma,
    namely in time $\poly{\log m,t}$ and space $\poly{\log m + \log
      t}$, with the extra minor restriction that $m > t^{\log
      t}$.}] \label{lem:optimalweakdesign} For every $t,m \in \N$ there
  exists a weak $(t,1)$-design $S_1,\dotsc,S_m \subset [d]$ such that
  $d = t \left\lceil \frac{t}{\ln 2} \right\rceil \left\lceil\log 4m
  \right\rceil = O(t^2 \log m)$. Moreover, such a design can be found
  in time $\poly{m,d}$ and space $\poly{m}$.
\end{lem}

As $1$-bit extractor, Raz et al.~\cite{RRV02} (and
Trevisan~\cite{Tre01} too) used the bits of a list-decodable code. We
give the parameters here as \propref{prop:codeextractor} and refer
to \appendixref{sec:codesRextractors} for details on the construction
and proof.

\begin{prop}
  \label{prop:codeextractor}
  For any $\eps > 0$ and $n \in \N$ there exists a $(k,\eps)$-strong
  extractor with uniform seed $\Ext_{n,\eps} : \{0,1\}^n \times
  \{0,1\}^t \to \{0,1\}$ with $t = O(\log (n / \eps))$ and $k = 3 \log
  1/\eps$.
\end{prop}

Plugging this into \thmref{thm:security} we get a quantum extractor
with parameters similar to Raz et al.~\cite{RRV02}.

\begin{cor}
  \label{cor:optimalentropyloss}
  Let $C : \{0,1\}^n \times \{0,1\}^t \to
  \{0,1\}$ be the extractor from Proposition~\ref{prop:codeextractor} with error $\eps' =
  \frac{\eps^2}{9m^2}$ and let $S_1,\dotsc,S_m \subset [d]$ be
  the weak $(t,1)$-design from
  \lemref{lem:optimalweakdesign}. Then
  \begin{align*}
    \Ext : \ & \{0,1\}^n  \times \{0,1\}^d \to \{0,1\}^m \\
     & (x,y) \mapsto C(x,y_{S_1}) \dotsb C(x,y_{S_m})
  \end{align*}
  is a quantum-proof $(m + 8 \log m + 8 \log 1/ \eps +
  O(1),\eps)$-strong extractor with uniform seed, with $d = O(\log^2
  (n/\eps) \log m)$.
\end{cor}

For $\eps = \poly{1/n}$ the seed has length $d = O(\log^3 n)$.
The entropy loss is $\Delta = 8 \log m + 8 \log 1/ \eps + O(1)$, which
means that the input still has this much randomness left in it
(conditioned on the output). We can extract a bit more by now applying
a second extractor to the input. For this we will use the extractor by
Tomamichel et al~\cite{TSSR10}, which is a quantum $(k',\eps')$-strong
extractor\footnote{\cite{TSSR10} define quantum-proof extractors a
  little differently than we do (see \footnoteref{fn:TSSRdef} on
  \pref{fn:TSSRdef}), but it is not hard to see that their result
  holds with the same parameters as the differences are absorbed in
  the $O$-notation.} with seed length $d' = O(m' +
\log n' + \log 1/\eps')$ and entropy loss $\Delta' = 4 \log 1/\eps' +
O(1)$, where $n'$ and $m'$ are the input and output string
lengths. Since we will use it for $m' = 8 \log m + 4 \log 1/ \eps' + O(1)$,
we immediately get the following corollary from
\lemref{lem:composition}.

\begin{cor}
  \label{cor:optimalentropyloss2}
  By applying the extractors from \corref{cor:optimalentropyloss} and
  \cite[Theorem 10]{TSSR10} in succession, we get a new function
  $\Ext: \{0,1\}^n \times \{0,1\}^d \to \{0,1\}^m$, which is a
  quantum-proof $(m + 4 \log 1/ \eps + O(1),\eps)$-strong extractor
  with uniform seed of length $d = O(\log^2 (n/\eps) \log m)$.
\end{cor}

For $\eps = \poly{1/n}$ the seed has length $d = O(\log^3 n)$.

The entropy loss is $\Delta = 4 \log 1/ \eps + O(1)$, which is only a
factor $2$ times larger than the optimal entropy loss.  By
\lemref{lem:smooth-min-entropy} this extractor can produce $m =
\HminSmooth{\eps'}{X|E} - 4 \log 1/ \eps - O(1)$ bits of randomness
with an error $\eps+2\eps'$.

\subsection{Seed of logarithmic size}
\label{subsec:logseed}

The weak design used in \secref{subsec:optimalentropyloss} requires a
seed of length $d = \Theta(t^2 \log m)$, where $t$ is the size of the
seed of the $1$-bit extractor. Since $t$ cannot be less than
$\Omega(\log n)$~\cite{RT00}, a scheme using this design will always
have $d = \Omega(\log^2 n \log m)$. If we want to use a seed of size
$d = O(\log n)$ we need a different weak design, e.g.,
\lemref{lem:logseedweakdesign}, at the cost of extracting less
randomness from the source.

\begin{lem}[\protect{\cite[Lemma 15]{RRV02}}]
  \label{lem:logseedweakdesign} For every $t,m \in \N$
  and $r > 1$, there exists a weak $(t,r)$-design $S_1,\dotsc,S_m
  \subset [d]$ such that $d = t \left\lceil t/\ln r
  \right\rceil = O\left(t^2/\log r \right)$. Moreover, such a
  design can be found in time $\poly{m,d}$ and space $\poly{m}$.
\end{lem}

For the $1$-bit extractor we can use the same as in the previous
section, \propref{prop:codeextractor}.

Plugging this into \thmref{thm:security} 
we get a quantum extractor with logarithmic seed length.

\begin{cor}
  \label{cor:logseed}
  If for any constant $0 < \alpha \leq 1$, the source has min-entropy
  $\Hmin{X|E} = n^{\alpha}$, and the desired error is $\eps =
  \poly{1/n}$, then using the extractor $C : \{0,1\}^n
  \times \{0,1\}^t \to \{0,1\}$ from \propref{prop:codeextractor} with error
  $\eps' = \frac{\eps^2}{9m^2}$ and the weak $(t,r)$-design
  $S_1,\dotsc,S_m \subset [d]$ from \lemref{lem:logseedweakdesign}
  with $r = n^\gamma$ for any $0 < \gamma < \alpha$, we have that
  \begin{align*}
    \Ext : \ & \{0,1\}^n  \times \{0,1\}^d \to \{0,1\}^m \\
     & (x,y) \mapsto C(x,y_{S_1}) \dotsb C(x,y_{S_m})
  \end{align*}
  is a quantum-proof $(n^{\gamma}m + 8 \log m + 8 \log 1/ \eps +
  O(1),\eps)$-strong extractor with uniform seed, with $d = O
  \left(\frac{1}{\gamma}\log n \right)$.
\end{cor}

Choosing $\gamma$ to be a constant results in a seed of length $d =
O(\log n)$. The output length is $m = n^{\alpha - \gamma} - o(1) =
\Hmin{X|E}^{1 - \frac{\gamma}{\alpha}} - o(1)$. By
\lemref{lem:smooth-min-entropy} this can be increased to $m =
\HminSmooth{\eps'}{X|E}^{1 - \frac{\gamma}{\alpha}} - o(1)$ with
an error of $\eps+2\eps'$.

\subsection{Locally computable extractor}
\label{subsec:local}

Another interesting feature of extractors is \emph{locality}, that
is, the $m$-bit output depends only a small subset of the $n$ input
bits. This is useful in, e.g., the bounded storage model (see
\cite{Mau92,Lu04,Vad04} for the case of a classical adversary and
\cite{KR07} for a general quantum treatment), where we assume a huge
source of random bits, say $n$, are available, and the adversary's
storage is bounded by $\alpha n$ for some constant $\alpha < 1$. Legitimate
parties are also assumed to have bounded workspace for computation. In
particular, for the model to be meaningful, the bound is stricter than
that on the adversary. So to extract a secret key from the large
source of randomness, they need an extractor which only reads $\ell
\ll n$ bits. An extractor with such a property is called
$\ell$-local.

\begin{deff}[$\ell$-local extractor]
  \label{def:localextractor}
  An extractor $\Ext: \{0,1\}^n \times \{0,1\}^d \to \{0,1\}^m$ is
  \emph{$\ell$-locally computable} (or \emph{$\ell$-local}), if for
  every $r \in \{0,1\}^d$, the function $x\mapsto \Ext(x,r)$ depends
  on only $\ell$ bits of its input, where the bit locations are
  determined by $r$.
\end{deff}

Lu~\cite{Lu04} modified Trevisan's scheme~\cite{Tre01,RRV02} to use a
local list-decodable code as $1$-bit extractor. Vadhan~\cite{Vad04}
proposes another construction for local extractors, which is optimal
up to constant factors. Both these constructions have similar
parameters in the case of $1$-bit extractors.\footnote{If the
  extractor is used to extract $m$-bits, then Vadhan's scheme reads
  less input bits and uses a shorter seed than Lu's.} We state the
parameters of Vadhan's construction here and refer the interested
reader to \cite{Lu04} for Lu's constructions.

\begin{lem}[\protect{\cite[Theorem 8.5]{Vad04}}]
  \label{lem:localextractor}
  For any $\eps > \exp \left( -n/2^{O(\log^* n)} \right)$, $n \in
  \N$ and constant $0 < \gamma < 1$, there exists an explicit
  $\ell$-local $(k,\eps)$-strong extractor with uniform seed
  $\Ext_{n,\eps,\gamma} : \{0,1\}^n \times \{0,1\}^t \to \{0,1\}$
  with $t = O(\log (n / \eps))$, $k = \gamma n$ and $\ell = O
  (\log 1/\eps)$.
\end{lem}

Since we assume that the available memory is limited, we also
want the construction of the weak design to be particularly
efficient. For this we can use a construction by Hartman and
Raz~\cite{HR03}.

\begin{lem}[\protect{\cite[Theorem 3]{HR03}}]
  \label{lem:localweakdesign}
  For every $m, t \in \N$, such that $m = \Omega (t^{\log t})$, and
  constant $r > 1$, there exists an explicit weak $(t,r)$-design
  $S_1,\dotsc,S_m \subset [d]$, where $d = O(t^2)$.  Such a design can
  be found in time $\poly{\log m,t}$ and space $\poly{\log m + \log
    t}$.
\end{lem}

\begin{rem}
  For the extractor from \lemref{lem:localextractor} and an error
  $\eps = \poly{1/n}$, this design requires $m = \Omega \left(
    (\log n)^{\log \log n} \right)$. If we are interested in a smaller
  $m$, say $m = \poly{\log n}$, then we can use the weak design from
  \lemref{lem:logseedweakdesign} with $r = n^{\gamma}$. This
  construction would require time and space $\poly{\log n} =
  \poly{\log 1/ \eps}$. The resulting seed would have length only
  $O(\log n)$ instead of $O(\log^2 n)$.
\end{rem}

Plugging these constructions into \thmref{thm:security} we get a
quantum local extractor.

\begin{cor}
  \label{cor:localextractor}
  If for any constant $0 < \alpha \leq 1$, the source has min-entropy
  $\Hmin{X|E} = \alpha n$, then using the weak $(t,r)$-design
  $S_1,\dotsc,S_m \subset [d]$ from \lemref{lem:localweakdesign} for
  any constant $r > 1$, and the extractor $C
  : \{0,1\}^n \times \{0,1\}^t \to \{0,1\}$ from
  \lemref{lem:localextractor} with error $\eps' = \frac{\eps^2}{9m^2}$ and
  any constant $\gamma < \alpha$, we have that
  \begin{align*}
    \Ext : \ & \{0,1\}^n  \times \{0,1\}^d \to \{0,1\}^m \\
     & (x,y) \mapsto C(x,y_{S_1}) \dotsb C(x,y_{S_m})
  \end{align*}
  is a quantum-proof $\ell$-local $(\gamma n + rm + 2 \log m + 2 \log
  1/ \eps + O(1),\eps)$-strong extractor with uniform seed, with $d =
  O(\log^2 (n/\eps))$ and $\ell = O(m \log (m/\eps))$. Furthermore,
  each bit of the output depends on only $O(\log (m /\eps))$ bits of
  the input.
\end{cor}

With these parameters the extractor can produce up to $m = (\alpha -
\gamma)n/r - O(\log 1/\eps)= (\Hmin{X|E} - \gamma n)/r - O(\log
1/\eps)$ bits of randomness, with an error of $\eps = \poly{1/n}$. By
\lemref{lem:smooth-min-entropy} this can be increased to $m =
(\HminSmooth{\eps'}{X|E} - \gamma n)/r - O(\log 1/\eps)$ with
an error of $\eps+2\eps'$.

\subsection{Weak random seed}
\label{subsec:weak-random-seed}

Extractors with weak random seeds typically require the seed to have a
min-entropy linear in its length. \thmref{thm:security.non-uniform}
says that the difference between the length and the min-entropy of the
seed needed in Trevisan's extractor is roughly the same as the
difference between the length and min-entropy of the seed of the
underlying $1$-bit extractor. So we will describe in detail how to
modify the construction from \secref{subsec:logseed} to use a weakly
random seed. As that extractor uses a seed of length $O(\log n)$, this
new construction allows us to preserve the linear loss in the
min-entropy of the seed. Any other version of Trevisan's extractor can
be modified in the same way to use a weakly random seed, albeit with
weaker parameters.

We will use a result by Raz~\cite{Raz05}, which shows how to transform
any extractor which needs a uniform seed into one which can work with
a weakly random seed.

\begin{lem}[\protect{\cite[Theorem 4]{Raz05}}]
  \label{lem:raz}
  For any $(k,\eps)$-strong extractor $\Ext : \{0,1\}^n \times
  \{0,1\}^t \to \{0,1\}^m$ with uniform seed, there exists a
  $(k,2\eps)$-strong extractor $\Ext : \{0,1\}^n \times \{0,1\}^{t'}
  \to \{0,1\}^m$ requiring only a seed with min-entropy $\Hmin{Y} \geq
  \left(\frac{1}{2} + \beta\right)t'$, where $t' = 8t/\beta$.
\end{lem}

By applying this lemma to the $1$-bit extractor given in
\propref{prop:codeextractor}, we obtain the following $1$-bit
extractor.

\begin{cor}
  \label{cor:razNcode}
  For any $\eps > 0$ and $n \in \N$ there exists a $(k,\eps)$-strong
  extractor $\Ext_{n,\eps} : \{0,1\}^n \times \{0,1\}^d \to \{0,1\}$
  requiring a seed with min-entropy $\left(\frac{1}{2}+\beta\right)d$, where
    $d = O(\frac{1}{\beta}\log (n / \eps))$ and $k = 3 \log 1/\eps+3$.
\end{cor}

Plugging this and the weak design from \lemref{lem:logseedweakdesign}
in \thmref{thm:security.non-uniform}, we get the following extractor
with weak random seed.

\begin{cor}
  \label{cor:weak-random-seed}
   Let $\alpha > 0$ be a constant such that the source has min-entropy
  $\Hmin{X|E} = n^{\alpha}$, and the desired error is $\eps =
  \poly{1/n}$. Using the extractor $C : \{0,1\}^n
  \times \{0,1\}^t \to \{0,1\}$ from \corref{cor:razNcode} with error
  $\eps' = \frac{\eps^2}{9m^2}$ and the weak $(t,r)$-design
  $S_1,\dotsc,S_m \subset [d]$ from \lemref{lem:logseedweakdesign}
  with $r = n^\gamma$ for any $0 < \gamma < \alpha$, we have that
  \begin{align*}
    \Ext : \ & \{0,1\}^n  \times \{0,1\}^d \to \{0,1\}^m \\
     & (x,y) \mapsto C(x,y_{S_1}) \dotsb C(x,y_{S_m})
  \end{align*}
  is a quantum-proof $(n^{\gamma}m + 8 \log m + 8 \log 1/ \eps +
  O(1),\eps)$-strong extractor with an $s$-bit weak random seed, where
  the seed has length $d = O\left(\frac{1}{\beta^2\gamma}\log n
  \right)$ and min-entropy $s = \left(1 -
    \frac{\frac{1}{2}-\beta}{c}\right)d$, for some constant
  $c$.\footnote{If we work out the exact constant, we find that $c
    \approx d/t \approx \frac{8(1+4a)}{\beta\gamma\ln 2}$, for $\eps =
    n^{-a}$.}
\end{cor}

Choosing $\beta$ and $\gamma$ to be constants results in a seed of
length $d = O(\log n)$ with a possible entropy-loss linear in
$d$. The output length is the same as in \secref{subsec:logseed}, $m =
n^{\alpha - \gamma} - o(1) = \Hmin{X|E}^{1 - \frac{\gamma}{\alpha}} -
o(1)$.

If we are interested in extracting all the min-entropy of the source,
we can combine \lemref{lem:raz} with the extractor from
\secref{subsec:optimalentropyloss}. This results in a new extractor
with seed length $d = O(\log^3 n)$ and seed min-entropy $s = d -
O(\sqrt[3]{d})$.

\section{Outlook}
\label{sec:others}

There exist many results modifying and improving Trevisan's
extractor. We briefly describe a few of them here, and refer to
\cite{Sha02} for a more extensive review.

Some of these constructions still follow the ``design and $1$-bit
extractor'' pattern --- hence our work implies that they are
immediately quantum-proof with roughly the same parameters --- e.g.,
the work of Raz et al.~\cite{RRV02} and Lu~\cite{Lu04}, which were
mentioned in \secref{sec:plugging} and correspond to modifications of
the design and $1$-bit extractor respectively. Other results such as
\cite{RRV02,TZS06,SU05} replace the binary list-decoding codes with
multivariate codes over a field $F$. Raz et al.~\cite{RRV02} use this
technique to reduce the dependence of the seed on the error from
$O(\log^2 1/\eps)$ to $O(\log 1/\eps)$. Ta-Shma et al.~\cite{TZS06}
and Shaltiel and Umans~\cite{SU05} reduce the size of the seed to $d
\leq 2 \log n$ in several constructions with different parameters for
the min-entropy. In these constructions the connection to $1$-bit
extractors is not clear anymore, and it is therefore not guaranteed
that these extractors are quantum-proof.

Raz et al.~\cite{RRV02} extract a little more randomness than we do in
\secref{subsec:optimalentropyloss}. They achieve this by composing (in
the sense described in \appendixref{subsec:compsing-extractors}) the scheme
of \corref{cor:optimalentropyloss} with an extractor by Srinivasan and
Zuckerman~\cite{SZ99}, which has an optimal entropy loss of $\Delta =
2 \log 1/\eps + O(1)$. In the presence of quantum side information
this extractor has been proven to have an entropy loss of $\Delta = 4
\log 1/\eps + O(1)$ in \cite{TSSR10}, hence our slightly weaker result
in \corref{cor:optimalentropyloss2}, which can possibly be improved.

Impagliazzo et al.~\cite{ISW00} and then Ta-Shma et al.~\cite{TUZ01}
modify Trevisan's extractor to work for a sub-polynomial entropy
source, still using a seed of size $d = O(\log n)$. Ta-Shma et
al.~\cite{TUZ01} achieve a construction which can extract all the
min-entropy $k$ of the source with such a seed length, for some $k =
o(n)$. While it is unclear whether these modifications preserve the
``design and $1$-bit extractor'' structure, it is an interesting open
problem to analyze them in the context of quantum side information.

Another research direction consists in making these constructions
practically implementable. Whether the extractor is used for privacy
amplification~\cite{BBR88,BBCM95}, generating true
randomness~\cite{XQMXZL11}, or for randomness recycling~\cite{IZ89},
the extractor has to have a running time which makes it useful. This
does not seem to be the case of Trevisan's
construction~\cite{Sol10}. An important open problem is thus to find
variations which are practical to execute.

It is also of great interest to study quantum-proof two-source
extractors, that is, extractors which can be applied to two
independent sources, each of which is correlated to independent
quantum side information. This has so far only been studied by Roy and
Kempe~\cite{KK10}, and we refer to their work for more details and
open problems.



\appendix
\appendixpage

\section{More on extractors}
\label{sec:more}

\subsection{Weak random seed}
\label{subsec:weak-seed}

In \secref{subsec:extractors.defs} we defined extractors as functions
which take a uniformly random seed. This is the most common way of
defining them, but not a necessary condition. Instead we can consider
extractors which use a seed which is only weakly random, but with
bounded min-entropy. We extend \defref{def:extractor} this way.
\begin{deff}[strong extractor with weak random seed \cite{Raz05}]
  \label{def:extractor.weak-seed}
  A function $\Ext: \{0,1\}^n \times \{0,1\}^d \to \{0,1\}^m$ is a
  \emph{$(k,\eps)$-strong extractor with an $s$-bit seed}, if for all
  distributions $X$ with $\Hmin{X} \geq k$ and any seed $Y$
  independent from $X$ with $\Hmin{Y} \geq s$, we have \[\frac{1}{2}
  \trnorm{ \rho_{\Ext(X,Y)Y} - \rho_{U_m} \tensor \rho_Y} \leq \eps, \]
  where $\rho_{U_m}$ is the fully mixed state on a system of dimension
  $2^m$.
\end{deff}

If quantum side information about the input is present in a system
$E$, then as before, we require the seed and the output to be
independent from that side-information.
\begin{deff}[quantum-proof strong extractor with weak random seed]
  \label{def:extractorwithadversary.weak-seed}
  A function $\Ext: \{0,1\}^n \times \{0,1\}^d \to \{0,1\}^m$ is a
  \emph{quantum-proof $(k,\eps)$-strong extractor with an $s$-bit
    seed}, if for all states $\rho_{XE}$ classical on $X$ with
  $\Hmin[\rho]{X|E} \geq k$, and for any seed $Y$ independent from
  $XE$ with $\Hmin{Y} \geq s$, we have \[\frac{1}{2}
  \trnorm{\rho_{\Ext(X,Y)YE} - \rho_{U_m} \tensor \rho_Y \tensor
    \rho_E} \leq \eps, \] where $\rho_{U_m}$ is the fully mixed
  state on a system of dimension $2^m$.
\end{deff}

\lemref{lem:KT08} says that any extractor will work with roughly the
same parameters when classical side information about the input $X$ is
present. The same holds in the case of classical side information
$Z$ about the seed $Y$.

\begin{lem}
  \label{lem:non-uniform-seed}
  Let $\Ext : \{0,1\}^n \times \{0,1\}^d \to \{0,1\}^m$ be a
  quantum-proof $(k,\eps)$-strong extractor with an $s$-bit seed. Then
  for any classical $X$, $Y$ and $Z$, and quantum $E$, such that $XE$
  and $Y$ are independent, $Y \leftrightarrow Z \leftrightarrow E$
  form a Markov chain,\footnote{A ccq state $\rho_{XYE}$ forms a
    Markov chain $X \leftrightarrow Y \leftrightarrow E$ if it can be
    expressed as $\rho_{XYE} = \sum_{x,y}
    P_{XY}(x,y) \proj{x,y} \tensor \rho_E^y.$} $\Hmin{Y|Z} \geq s +
  \log 1/\eps$, and for all $z \in \cZ$, $\Hmin{X|EZ=z} \geq k$, we have
  \[ \frac{1}{2} \trnorm{ \rho_{\Ext(X,Y)YZE} - \rho_{U} \tensor
    \rho_{YZE}} \leq 2\eps.\]
\end{lem}

\begin{proof}  For any two classical systems $Y$ and $Z$, we
have \[2^{-\Hmin{Y|Z}} = \E_{z \leftarrow Z}
\left[2^{-\Hmin{Y|Z=z}}\right],\] so by Markov's inequality, \[\Pr_{z
  \leftarrow Z} \left[ \Hmin{Y|Z=z} \leq \Hmin{Y|Z} - \log 1/\eps
\right] \leq \eps.\] And since $Y \leftrightarrow Z \leftrightarrow E$
form a Markov chain, we have for all $z \in \cZ$, \[ \rho_{YE|Z=z} =
\rho_{Y|Z=z} \tensor \rho_{E|Z=z}.\] Hence
  \begin{align*}
    & \frac{1}{2} \trnorm{\rho_{\Ext(X,Y)YEZ} - \rho_{U} \tensor
      \rho_{YEZ}} \\ & \qquad \qquad = \frac{1}{2} \sum_{z \in \cZ} P_Z(z)
    \trnorm{\rho_{\Ext(X,Y)YE|Z=z} - \rho_{U} \tensor
      \rho_{YE|Z=z}} \\
    & \qquad \qquad = \frac{1}{2} \sum_{z \in \cZ} P_Z(z)
    \trnorm{\rho_{\Ext(X,Y)YE|Z=z} - \rho_{U} \tensor \rho_{Y|Z=z}
      \tensor \rho_{E|Z=z}} \leq 2\eps. \qedhere
  \end{align*}\end{proof}

The case of quantum side information correlated to both the
input and the seed is out of the scope of this work.

\subsection{Composing extractors}
\label{subsec:compsing-extractors}

If an extractor does not have optimal entropy loss, a useful approach
to extract more entropy is to apply a second extractor to the original
input, to extract the randomness that remains when the output
of the first extractor is known. This was first proposed in the
classical case by Wigderson and Zuckerman~\cite{WZ99}, and improved by
Raz et al.~\cite{RRV02}. K\"onig and Terhal~\cite{KT08} gave the first
quantum version for composing $m$ times quantum $1$-bit extractors. We
slightly generalize the result of K\"onig and Terhal~\cite{KT08} to
the composition of arbitrary quantum extractors.

\begin{lem}
  \label{lem:composition}
  Let $\Ext_1: \{0,1\}^n \times \{0,1\}^{d_1} \to \{0,1\}^{m_1}$ and
  $\Ext_2: \{0,1\}^n \times \{0,1\}^{d_2} \to \{0,1\}^{m_2}$ be
  quantum-proof $(k,\eps_1)$- and $(k-m_1,\eps_2)$-strong
  extractors. Then the composition of the two, namely
  \begin{align*}
    \Ext_3 : & \{0,1\}^n \times \{0,1\}^{d_1} \times \{0,1\}^{d_2} \to
    \{0,1\}^{m_1} \times \{0,1\}^{m_2} \\
    & (x,y_1,y_2) \mapsto (\Ext_1(x,y_1),\Ext_2(x,y_2)),
  \end{align*}
  is a quantum-proof $(k,\eps_1+\eps_2)$-strong extractor.
\end{lem}

\begin{proof}
  We need to show that for any state $\rho_{XE}$ with $\Hmin{X|E} \geq
  k$, \begin{equation} \label{eq:composition.1} \frac{1}{2} \trnorm{
      \rho_{\Ext_1(X,Y_1)\Ext_2(X,Y_2) Y_1Y_2E} - \rho_{U_1} \tensor
      \rho_{U_2}\tensor \rho_{Y_1} \tensor \rho_{Y_2} \tensor \rho_E}
    \leq \eps_1 + \eps_2.\end{equation}
  The left-hand side of \eqnref{eq:composition.1} can be upper-bounded
  by
  \begin{multline} \label{eq:composition.2} \frac{1}{2} \trnorm{
      \rho_{\Ext_1(X,Y_1)Y_1E} \tensor \rho_{U_2} \tensor \rho_{Y_2} -
      \rho_{U_1} \tensor \rho_{Y_1} \tensor \rho_E \tensor
      \rho_{U_2}\tensor \rho_{Y_2} } \\+ \frac{1}{2} \trnorm{
      \rho_{\Ext_2(X,Y_2)Y_2 \Ext_1(X,Y_1)Y_1E} - \rho_{U_2} \tensor
      \rho_{Y_2} \tensor \rho_{\Ext_1(X,Y_1)Y_1E}}.\end{multline} By
  the definition of $\Ext_1$ the first term in
  \eqnref{eq:composition.2} is upper-bounded by $\eps_1$.  For the
  second term we use \lemref{lem:min-entropy.rule.3} and
  get \begin{multline*} \Hmin{X|\Ext_1(X,Y_1)Y_1E} \geq \Hmin{X|Y_1E}
    - H_0(\Ext_1(X,Y_1)) \\ = \Hmin{X|E} - H_0(\Ext_1(X,Y_1)) \geq k -
    m_1.\end{multline*} By the definition of $\Ext_2$ the second term
  in \eqnref{eq:composition.2} can then be upper-bounded by
  $\eps_2$. \end{proof}

\section{Technical lemmas}
\label{sec:lemmas}

\subsection{Min-entropy chain rules}

We use the following ``chain-rule type'' statement about the
min-entropy. The proofs for the two first can be found in
\cite{Ren05}.

\begin{lem}[\protect{\cite[Lemma 3.1.10]{Ren05}}]
  \label{lem:min-entropy.rule.1}
  For any state $\rho_{ABC}$,
  \[\Hmin{A|BC} \geq \Hmin{AC|B} - H_0(C),\] where $H_0(C) = \log
  \rank{\rho_C}$.
\end{lem}

\begin{lem}[\protect{\cite[Lemma 3.1.9]{Ren05}}]
  \label{lem:min-entropy.rule.2}
  For any state $\rho_{ABZ}$ classical on $Z$, \[\Hmin{AZ|B} \geq
  \Hmin{A|B}.\]
\end{lem}

\begin{lem}
  \label{lem:min-entropy.rule.3}
  For any state $\rho_{ABZ}$ classical on $Z$, \[\Hmin{A|BZ} \geq \Hmin{A|B} - H_0(Z),\]
  where $H_0(Z) = \log \rank{\rho_Z}$.
\end{lem}

\begin{proof}
  Immediate by combining \lemref{lem:min-entropy.rule.1} and
  \lemref{lem:min-entropy.rule.2}.
\end{proof}

\subsection{Reduction step}
\label{sec:security.reduction}

To show that a player who can distinguish the output of $\Ext_C$
(defined in \defref{def:genericscheme} on \pref{def:genericscheme})
from uniform can also guess the output of the extractor $C$, we first
show that such a player can guess one of the bits of the output of
$\Ext_C$ given some extra classical information. This is a quantum
version of a result by Yao~\cite{Yao82}.

\begin{lem} \label{lem:Yao's-lem} Let $\rho_{ZB}$ be a cq-state,
  where $Z$ is a random variable on $m$-bit strings. If
  $\trnorm{\rho_{ZB} - \rho_{U_m} \tensor \rho_B} > \eps$, then there
  exists an $i \in [m]$ such that
  \begin{equation} \label{eq:Yao's-lem}
    \trnorm{\sum_{\substack{z \in \cZ \\ z_{i} = 0}} p_z
    \proj{z_{[i-1]}} \tensor \rho^z_B - \sum_{\substack{z \in \cZ \\ z_{i}
        = 1}} p_z \proj{z_{[i-1]}} \tensor \rho^z_B } >
  \frac{\eps}{m}.
  \end{equation}
\end{lem}

Using the fact that for any \emph{binary} random
    variable $X$ and quantum system $Q$ with $\rho_{XQ} =
    \sum_{i=0,1} p_i \proj{i} \tensor \rho^i_Q$, the following equality holds:
  $\trnorm{\rho_{XQ} - \rho_{U_1} \tensor \rho_Q } = \trnorm{p_0
    \rho^0_Q - p_1 \rho^1_Q}$, \eqnref{eq:Yao's-lem} can be
  rewritten as $\trnorm{\rho_{Z_{i[i-1]}B} - \rho_{U_1} \tensor
    \rho_{Z_{[i-1]}B}} > \frac{\eps}{m}$. \lemref{lem:Yao's-lem} can thus be
  interpreted as saying that if a player holding $B$ can distinguish
  $Z$ from uniform with probability greater than $\eps$, then there
  exists a bit $i \in [m]$ such that when given the previous $i-1$
  bits of $Z$, he can distinguish the \ith{i} bit of $Z$ from uniform
  with probability greater than $\frac{\eps}{m}$.

\begin{proof}
  The proof uses a hybrid argument. Let \[\sigma_i =
  \sum_{\substack{z \in \cZ \\ r \in \{0,1\}^m}} \frac{p_z}{2^m}
  \proj{z_{[i]}, r_{\{i+1,\dotsc,m\}}} \tensor \rho^z_B.\] Then
  \begin{align*}
    \eps & < \trnorm{\rho_{ZB} - \rho_{U_m} \tensor \rho_B} \\
      & = \trnorm{\sigma_m - \sigma_0} \\
      & \leq \sum_{i = 1}^m \trnorm{\sigma_i - \sigma_{i-1}} \\
      & \leq m \max_i  \trnorm{\sigma_i - \sigma_{i-1}}.
  \end{align*}

  By rearranging $\trnorm{\sigma_i - \sigma_{i-1}}$ we get the lhs of
  \eqnref{eq:Yao's-lem}.
\end{proof}

We now need to bound the size of this extra information, the
``previous $i-1$ bits'', and show that when averaging over all the
seeds of $\Ext_C$, we average over all the seeds of $C$, which means
that guessing a bit of the output of $\Ext_C$ corresponds to
distinguishing the output of $C$ from uniform. For the reader's 
convenience we now restate~\propref{prop:guessing-cx} and give its proof.

\begin{prop} \label{prop:guessing-cx-bis}[\propref{prop:guessing-cx}]
  Let $X$ be a classical random variable correlated to some quantum
  system $E$, let $Y$ be a (not necessarily uniform) seed, independent
  from $XE$, and let
  \begin{equation} \label{eq:guessing-cx.lem.1} \trnorm{\rho_{\Ext_C(X,Y)E}
      - \rho_{U_m} \tensor \rho_Y \tensor \rho_E} > \eps,
  \end{equation} where $\Ext_C$ is the extractor from
  \defref{def:genericscheme}. Then there exists a fixed partition of the
  seed $Y$ in two substrings $V$ and $W$, and a classical random
  variable $G$, such that $G$ has size $H_0(G) \leq rm$,
  where $r$ is one of the parameters of the weak design
  (\defref{def:weakdesign}), $V \leftrightarrow W
  \leftrightarrow G$ form a Markov chain, and
  \begin{equation} \label{eq:guessing-cx.lem.2}
     \trnorm{\rho_{C(X,V)VWGE} - \rho_{U_1} \tensor
         \rho_{VWGE}} > \frac{\eps}{m}.
   \end{equation}
\end{prop}


\begin{proof}
  We apply \lemref{lem:Yao's-lem} to \eqnref{eq:guessing-cx.lem.1} and get
  that there exists an $i \in [m]$ such that
  \begin{multline} \label{eq:guessing-cx.1}
      \left\| \sum_{\substack{x,y \\ C(x,y_{S_i}) =
          0}} p_xq_y \proj{ C(x,y_{S_1}) \dotsb C(x,y_{S_{i-1}} ),y}
      \tensor \rho^x \right. \\ - \left. \sum_{\substack{x,y \\
          C(x,y_{S_i}) = 1}} p_xq_y \proj{ C(x,y_{S_1}) \dotsb
        C(x,y_{S_{i-1}} ),y} \tensor \rho^x \right\|_{\text{tr}} \\
    > \frac{\eps}{m},
  \end{multline} where $\{p_x\}_{x \in \cX}$ and $\{q_y\}_{y \in \cY}$ are the probability
  distributions of $X$ and $Y$ respectively.

  We split $y \in \{0,1\}^d$ in two strings of $t = |S_i|$ and $d -
  t$ bits, and write $v \coloneqq y_{S_i}$ and $w \coloneqq y_{[d]
    \setminus S_i}$. To simplify the notation, we set $g(w,x,j,v)
  \coloneqq C(x,y_{S_j})$. Fix $w$, $x$ and $j$, and consider the
  function $g(w,x,j,\cdot) : \{0,1\}^t \to \{0,1\}$. This function
  only depends on $|S_j \cap S_i|$ bits of $v$. So to describe this
  function we need a string of at most $2^{|S_j \cap S_i|}$ bits. And to
  describe $g^{w,x}(\cdot) \coloneqq g(w,x,1,\cdot) \dotsb
  g(w,x,i-1,\cdot)$, which is the concatenation of the bits of
  $g(w,x,j,\cdot)$ for $1 \leq j \leq i-1$, we need a string of length
  at most $\sum_{j=1}^{i-1} 2^{|S_j \cap S_i|}$. So a system $G$
  containing a description of $g^{w,x}$ has size at most $H_0(G) \leq
  \sum_{j=1}^{i-1} 2^{|S_j \cap S_i|}$. We now rewrite
  \eqnref{eq:guessing-cx.1} as
  \begin{multline*}
    \left\| \sum_{\substack{x,v,w \\ C(x,v) = 0}}
      p_xq_{v,w}  \proj{g^{w,x}(v),v,w}
      \tensor \rho^x \right. \\ \left. - \sum_{\substack{x,v,w \\
          C(x,v) = 1}} p_xq_{v,w} \proj{g^{w,x}(v),v,w} \tensor \rho^x
    \right\|_{\text{tr}} > \frac{\eps}{m}.
  \end{multline*}

  By providing a complete description of $g^{w,x}$ instead of its
  value at the point $v$, we can only increase the trace distance,
  hence
  \begin{multline*}
    \left\| \sum_{\substack{x,v,w \\ C(x,v) = 0}}
      p_xq_{v,w} \proj{g^{w,x},v,w}
      \tensor \rho^x \right. \\ \left. - \sum_{\substack{x,v,w \\
          C(x,v) = 1}} p_xq_{v,w} \proj{g^{w,x},v,w} \tensor \rho^x
    \right\|_{\text{tr}} > \frac{\eps}{m}.
  \end{multline*}
  
  By rearranging this a little more we finally get
  \[ \trnorm{\rho_{C(X,V)VWGE} - \rho_{U_1} \tensor \rho_{VWGE}} >
  \frac{\eps}{m},\] where $G$ is a classical system of size $H_0(G)
  \leq \sum_{j=1}^{i-1} 2^{|S_j \cap S_i|}$, and $V \leftrightarrow W
  \leftrightarrow G$ form a Markov chain. By the definition of weak
  designs, we have for all $i \in [m]$, $\sum_{j=1}^{i-1} 2^{|S_j \cap
    S_i|} \leq rm$ for some $r \geq 1$. So $H_0(G) \leq rm$.
\end{proof}

\section{List-decodable codes are one-bit extractors}
\label{sec:codesRextractors}

A standard error correcting code guarantees that if the error is
small, any string can be uniquely decoded. A list-decodable code
guarantees that for a larger (but bounded) error, any string can be
decoded to a list of possible messages.

\begin{deff}[list-decodable code \cite{Sud00}]
  A code $C : \{0,1\}^n \to \{0,1\}^{\bar{n}}$ is said to be
  $(\eps,L)$-list-decodable if every Hamming ball of relative
  radius $1/2 - \eps$ in $\{0,1\}^{\bar{n}}$ contains at most
  $L$ codewords.
\end{deff}

Neither Trevisan~\cite{Tre01} nor Raz et al.~\cite{RRV02} state it
explicitly, but both papers contain
an implicit proof that if $C : \{0,1\}^n \to \{0,1\}^{\bar{n}}$ is a
$(\eps,L)$-list-decodable code, then \begin{align*}
    \Ext : \{0,1\}^n \times [\bar{n}] & \to \{0,1\} \\
      (x,y) & \mapsto C(x)_y,
  \end{align*}
  is a $(\log L + \log 1/2\eps, 2 \eps)$-strong extractor (according
  to \defref{def:extractor}). We have rewritten their proof as
  \thmref{thm:codesRextractors} for completeness.\footnote{A slightly
    more general proof, that \emph{approximate} list-decodable codes
    are $1$-bit extractors can be found in~\cite[Claim 3.7]{DV10}.}

There exist list-decodable codes with following parameters.



\begin{lem} \label{lem:ecc} For every $n \in \N$ and $\delta > 0$
  there is a code $C_{n,\delta} : \{0,1\}^n \to \{0,1\}^{\bar{n}}$,
  which is $(\delta,1/\delta^2)$-list-decodable, with $\bar{n} = \poly{
    n,1/\delta}$. Furthermore, $C_{n,\delta}$ can be evaluated in time
  $\poly{n,1/\delta}$ and $\bar{n}$ can be assumed to be a power of
  $2$.
\end{lem}

For example, Guruswami et al.~\cite{GHSZ02} combine a Reed-Solomon
code with a Ha\-da\-mard code, obtaining such a list-decodable code with
$\bar{n} = O(n/\delta^4)$.

Such codes require all bits of the input $x$ to be read to compute any
single bit $C(x)_i$ of the output. If we are interested in so-called
$\emph{local}$ codes, we can use a construction by Lu~\cite[Corollary 1]{Lu04}.


\begin{thm}
  \label{thm:codesRextractors}
  Let $C : \{0,1\}^n \to \{0,1\}^{\bar{n}}$ be an
  $(\eps,L)$-list-de\-co\-dable code. Then the function 
  \begin{align*}
    C' : \{0,1\}^n \times [\bar{n}] & \to \{0,1\} \\
      (x,y) & \mapsto C(x)_y,
  \end{align*}
  is a $(\log L + \log 1/2\eps, 2 \eps)$-strong
  extractor.\footnote{This theorem still holds in the presence of
    classical side information with exactly the same parameters.}
\end{thm}

To prove this theorem we first show that a player who can
distinguish the bit of $C'(X,Y)$ from uniform can construct a string
$\alpha$ which is close to $C(X)$ on average (over $X$). Then using
the error correcting properties of the code $C$, he can reconstruct
$X$. Hence a player who can break the extractor must have low
min-entropy about $X$.

\begin{lem}
  \label{lem:c-approximating-cx} Let $X$ and $Y$ be two independent
  random variables with alphabets $\{0,1\}^n$ and $[n]$
  respectively. Let $Y$ be uniformly distributed and $X$ be
  distributed such that $\frac{1}{2}| X_Y \circ Y - U_1 \circ Y| >
  \delta$, where $U_1$ is uniformly distributed on $\{0,1\}$. Then
  there exists a string $\alpha \in \{0,1\}^n$ with \[ \Pr
  \left[d(X,\alpha) \leq \frac{1}{2} - \frac{\delta}{2} \right] >
  \delta, \] where $d(\cdot,\cdot)$ is the relative Hamming distance.
\end{lem}

\begin{proof}
  Define $\alpha \in \{0,1\}^n$ to be the concatenation of the most
  probable bits of $X$, i.e., $\alpha_y \coloneqq \argmax_{b} P_{X_y}(b)$,
  where $P_{X_y}(b) = \sum_{\substack{x \in \{0,1\}^n \\ x_y = b}}
  P_X(x)$. 

  The average relative Hamming distance between $X$ and $\alpha$ is
  \begin{align*}
    \sum_{x \in \{0,1\}^n} P_X(x) d(x,\alpha) & = \frac{1}{n} \sum_{x \in \{0,1\}^n}
    P_X(x) \sum_{y =1}^n |x_y - \alpha_y| \\
    & = \frac{1}{n} \sum_{\substack{x,y \\ x_y \neq \alpha_y}} P_X(x)  =  1 -  \frac{1}{n} \sum_{y = 1}^n P_X(\alpha_y).
  \end{align*}

  And since $\frac{1}{2}| X_Y \circ Y - U_1
  \circ Y| > \delta$ is equivalent to $\frac{1}{n} \sum_{y = 1}^n
  \max_{b \in \{0,1\}} P_{X_y}(b) > \frac{1}{2} + \delta$, we have
  \begin{equation} \label{eq:c-lem.1} \sum_{x \in \{0,1\}^n} P_X(x)
    d(x,\alpha) < \frac{1}{2} - \delta. \end{equation}

  We now wish to lower bound the probability that the relative Hamming
  distance is less than $\frac{1}{2}-\frac{\delta}{2}$. Let $B
  \coloneqq \{x : d(x,\alpha) \leq \frac{1}{2} - \frac{\delta}{2}$\}
  be the set of values $x \in \{0,1\}^n$ meeting this
  requirement. Then the weight of $B$, $w(B) \coloneqq \sum_{x \in B}
  P_{X}(x)$, is the quantity we wish to lower bound. It is at its
  minimum if all $x \in B$ have Hamming distance $d(x,\alpha) = 0$. In
  which case the average Hamming distance is
  \begin{equation} \label{eq:c-lem.2} \sum_{x \in \{0,1\}^n} P_{X}(x)
    d \left( x,\alpha \right) > (1- w(B)) \left( \frac{1}{2} -
      \frac{\delta}{2} \right). \end{equation} Combining
  \eqnsref{eq:c-lem.1} and \eqref{eq:c-lem.2} we get
  \[ w(B) > \frac{\delta}{1-\delta} \geq \delta. \qedhere \] \end{proof}

We are now ready to prove \thmref{thm:codesRextractors}.

\begin{proof}[Proof of \thmref{thm:codesRextractors}.]
  We will show that if it is possible to distinguish $C'(X,Y)$ from
  uniform with probability at least $2 \eps$, then $X$ must have
  min-entropy $\Hmin{X} < \log L + \log 1/2\eps$.

  If $\frac{1}{2} \left| C'(X,Y) \circ Y - U_1 \circ Y
 \right| > 2 \eps$, then by \lemref{lem:c-approximating-cx} we
 know that there exists an $\alpha \in \{0,1\}^{\bar{n}}$ such that \[ \Pr \left[ d\left( C(X), \alpha \right)
  \leq \frac{1}{2} - \eps \right] > 2\eps,\]
where $d(\cdot,\cdot)$ is the relative Hamming distance.

This means that with probability at least $2\eps$, $X$ 
takes values $x$ such that the relative Hamming distance $d(C(x),\alpha) \leq
\frac{1}{2} - \eps$. So for these values of $X$, if we choose one
of the codewords in the Hamming ball of relative radius $\frac{1}{2} -
\eps$ around $\alpha$ uniformly at random as our guess for $x$, we will
have chosen correctly with probability at least $1/L$, since the
Hamming ball contains at most $L$ code words. The total probability of
guessing $X$ is then at least $2 \eps/L$.

Hence by \eqnref{eq:hmin=pguess}, $\Hmin{X} < \log L + \log
1/2\eps$.
\end{proof}

\newcommand{\etalchar}[1]{$^{#1}$}
\providecommand{\bibhead}[1]{}
\expandafter\ifx\csname pdfbookmark\endcsname\relax%
  \providecommand{\tocrefpdfbookmark}{}
\else\providecommand{\tocrefpdfbookmark}{%
   \hypertarget{tocreferences}{}%
   \pdfbookmark[1]{References}{tocreferences}}%
\fi

\tocrefpdfbookmark

\end{document}